\documentclass[final,3p,times,twocolumn,authoryear]{elsarticle}

\usepackage{lipsum}
\usepackage{amsmath,amssymb,amsfonts,amsthm}
\usepackage{algorithmic}
\usepackage{graphicx}   

\newtheorem{theorem}{Theorem}

\theoremstyle{remark}
\newtheorem{remark}{Remark}

\newtheorem{definition}{Definition}

\journal{}

\begin{document}

\begin{frontmatter}



\title{A Time-Barrier Lyapunov Condition for Predefined-Time Stability}


\author[first]{Özhan Bingöl}
\affiliation[first]{organization={The Department of Electrical and Electronic Engineering, Gumushane University},
            addressline={Baglarbasi Mah.}, 
            city={Gumushane},
            postcode={29100}, 
            country={Turkey}, \ead{ozhan.bingol@gumushane.edu.tr}}

\begin{abstract}
Predefined-time stability enables convergence within a user-specified time independent of initial conditions. Existing results are predominantly based on autonomous Lyapunov inequalities, where the predefined-time is realized through integral bounds on state-dependent decay and therefore acts as an upper bound rather than a structurally enforced deadline. This paper introduces a time-barrier predefined-time stability concept in which convergence is enforced through a nonautonomous Lyapunov mechanism that intrinsically restricts the remaining available time. A sufficient Lyapunov-based condition is established, guaranteeing convergence before the predefined deadline via divergence of a time-dependent barrier. It is further shown that this mechanism cannot be reproduced by classical autonomous predefined-time stability formulations, thereby constituting a distinct stability notion. The proposed approach provides a concise and transparent means of enforcing hard convergence deadlines in nonlinear systems.
\end{abstract}



\begin{keyword}
Predefined-time stability \sep nonautonomous systems \sep Lyapunov methods



\end{keyword}

\end{frontmatter}




\section{Introduction}
\label{sec1}

The explicit enforcement of convergence time is a central problem in nonlinear stability theory, particularly for systems operating under strict temporal or safety constraints. Classical Lyapunov theory guarantees asymptotic stability but provides no information on transient duration. Finite-time \cite{bhat2000finite} and fixed-time \cite{polyakov2011nonlinear} stability address this limitation by ensuring convergence within a finite interval, with fixed-time formulations further removing dependence on initial conditions. However, in fixed-time systems the convergence bound is implicit and determined by system parameters, preventing the direct assignment of an exact settling time.

Predefined-time stability was introduced to overcome this limitation by allowing the convergence time to be predefined explicitly and independently of initial conditions \cite{sanchez2015predefined,sanchez2018class,jimenez2020lyapunov}. Most existing predefined-time stability results rely on autonomous Lyapunov dissipation inequalities of the form $\dot V \le -\Phi(V)$, where the predefined-time enters through integral bounds on state decay. Although mathematically rigorous, such constructions realize the predefined time as an upper bound derived from gain- and exponent-dependent conditions. As a consequence, the convergence deadline is not enforced structurally but emerges indirectly through state-dependent decay shaping.

More recent arbitrary-time \cite{pal2020design} and prescribed-time \cite{song2017time} approaches incorporate time explicitly into the Lyapunov dynamics, enabling exact convergence at a user-specified instant. These methods typically rely on nonautonomous transformations, exponential time scaling, or trajectory-matching conditions. While effective, they enforce convergence by prescribing the temporal evolution of the state, which may obscure the underlying stability mechanism and require additional structural assumptions.

These observations raise a fundamental question: can predefined-time convergence be enforced intrinsically as a hard temporal constraint, rather than being obtained through accelerated state decay or trajectory matching? In particular, is it possible to guarantee convergence before a predefined deadline by restricting the remaining available time itself, independently of the state magnitude?

This paper addresses this question by introducing a time-barrier predefined-time stability concept based on a nonautonomous Lyapunov mechanism. Instead of shaping state-dependent decay rates, the proposed approach embeds a time-dependent barrier into the Lyapunov dissipation inequality, whose divergence as the predefined-time is approached prevents solution trajectories from persisting away from the equilibrium. As a result, convergence is enforced through intrinsic temporal infeasibility rather than aggressive state acceleration.

A sufficient Lyapunov-based condition is established that guarantees convergence before a user-specified time via divergence of the time barrier. It is further shown that this mechanism cannot be reproduced by classical autonomous predefined-time stability inequalities, thereby constituting a distinct stability notion rather than a reformulation of existing results.

\section{Problem Setting and Motivation}

Consider the nonlinear time-varying system
\begin{equation}
\dot{x}(t) = f(x(t),t), \qquad x(0)=x_0,
\label{eq1}
\end{equation}

where $x(t)\in\mathbb{R}^n$ and $f:\mathbb{R}^n\times\mathbb{R}_{\ge 0}\to\mathbb{R}^n$ is locally Lipschitz in $x$ and piecewise continuous in $t$. The origin is assumed to be an equilibrium point, i.e., $f(0,t)=0$ for all $t\ge 0$. Solutions are understood in the classical sense when $f$ is smooth and in the Filippov sense otherwise. Predefined-time stability is typically established using autonomous Lyapunov dissipation inequalities of the form \cite{sanchez2015predefined}

\begin{equation}
\dot V(x) \le -\Phi(V(x)),
\label{eq2}
\end{equation}

where $V:\mathbb{R}^n\to\mathbb{R}_{\ge 0}$ is positive definite and $\Phi(\cdot)$ is selected to satisfy an integral condition

\begin{equation}
\int_0^{V_0} \frac{dV}{\Phi(V)} \le T_c,
\label{eq3}
\end{equation}

uniformly for all admissible initial values $V_0$. In this construction, the predefined-time $T_c$ arises as an upper bound determined by the shape of $\Phi(\cdot)$ and the associated gain and exponent selections.

Although such autonomous formulations provide explicit control over convergence time bounds, the deadline itself is not enforced structurally. Instead, convergence before $T_c$ is achieved through sufficiently aggressive state-dependent decay, and the predefined- time emerges indirectly from integral boundedness requirements. As a result, the convergence mechanism remains entirely state-driven, and the settling time is realized as a conservative upper bound rather than as an intrinsic temporal constraint.

Nonautonomous approaches such as arbitrary-time \cite{pal2020design} and prescribed-time \cite{song2017time} stability incorporate time explicitly into the Lyapunov dynamics to achieve exact convergence at a user-defined instant. These methods typically enforce convergence by matching a predefined decay trajectory or applying time-scaling transformations. While effective, they rely on regulating the state evolution to satisfy a predefined temporal profile rather than on restricting temporal feasibility itself.

Motivated by these limitations, this paper adopts a different perspective on predefined-time convergence. Instead of accelerating state decay or prescribing a reference trajectory, convergence is enforced by progressively restricting the remaining available time through a time-dependent barrier. In such a setting, solution trajectories are prevented from persisting away from the equilibrium as the predefined deadline is approached, regardless of their initial condition or decay rate. This principle forms the basis of the time-barrier predefined-time stability notion introduced in the next section.

\section{Time-Barrier Predefined-Time Stability}

This section introduces the proposed time-barrier predefined-time stability notion and establishes a sufficient Lyapunov-based condition guaranteeing convergence before a predefined deadline. Unlike classical autonomous predefined-time formulations, the convergence mechanism is enforced through a nonautonomous time-dependent barrier that intrinsically restricts the remaining available time.

\begin{definition}[Time-Barrier Predefined-Time Stability]
\label{def1}

Consider system \eqref{eq1} and a predefined-time $T_c>0$. The equilibrium $x=0$ is said to be \emph{time-barrier predefined-time stable} if, for every initial condition $x_0$, the corresponding solution $x(t;x_0)$ exists on $[0,T_c)$ and satisfies

\begin{equation}
\lim_{t \to T_c^-} x(t;x_0) = 0.
\label{eq4}
\end{equation}

\end{definition}

Definition \ref{def1} characterizes convergence through a hard temporal boundary: solution trajectories are not permitted to remain away from the equilibrium as the predefined deadline is approached. In contrast to classical predefined-time stability, convergence is not quantified via integral bounds on state decay but is instead enforced by the behavior near the terminal time.

\begin{theorem}[Time-Barrier Predefined-Time Stability Condition]
\label{thm1}

Let $V:\mathbb{R}^n\to\mathbb{R}_{\ge 0}$ be a positive definite and locally Lipschitz function. Assume that for every initial condition $x_0$, a solution of system \eqref{eq1} exists on $[0,T_c)$. Suppose that for all $x\neq 0$ and all $t\in[0,T_c)$,

\begin{equation}
\dot V(x(t),t) \le -\beta \frac{V(x(t),t)}{T_c-t} - q V(x(t),t)^{\alpha},
\label{eq5}
\end{equation}

where $q>0$, $\alpha\in(0,1)$, and $\beta>0$ satisfy

\begin{equation}
\beta(1-\alpha)\ge 1.
\label{eq6}
\end{equation}

Then the equilibrium $x=0$ is time-barrier predefined-time stable in the sense of Definition \ref{def1}.
\end{theorem}

\begin{proof}
Define the time-barrier Lyapunov transformation

\begin{equation}
W(t)=\frac{V(x(t),t)}{(T_c-t)^{\beta}}, \qquad t\in[0,T_c).
\label{eq7}
\end{equation}

Since $V\ge0$ and $T_c-t>0$, it follows that $W(t)\ge0$. Differentiating \eqref{eq7} along trajectories of \eqref{eq1} yields

\begin{equation}
\dot W(t)=\frac{\dot V(x(t),t)}{(T_c-t)^{\beta}}+\beta\frac{V(x(t),t)}{(T_c-t)^{\beta+1}}.
\end{equation}

Substituting \eqref{eq5} gives

\begin{equation}
\dot W(t)\le -q W(t)^{\alpha}(T_c-t)^{-\beta(1-\alpha)}.
\end{equation}

For $\alpha\in(0,1)$,

\begin{equation}
\frac{d}{dt}W(t)^{1-\alpha}\le -q(1-\alpha)(T_c-t)^{-\beta(1-\alpha)}.
\end{equation}

Integrating over $[0,t)$ and using condition \eqref{eq6}, the integral diverges as $t\to T_c^-$. Consequently, there exists $\tau\le T_c$ such that $W(\tau)=0$, which implies $V(x(\tau),\tau)=0$ and therefore

\begin{equation}
\lim_{t\to T_c^-}x(t)=0.
\end{equation}

\end{proof}

\begin{remark}
The time-barrier predefined-time stability condition in Theorem \ref{thm1} enforces convergence through intrinsic temporal infeasibility rather than state acceleration. The predefined time $T_c$ acts as a hard convergence deadline that is independent of the initial condition. Moreover, $T_c$ is not tightly coupled to the gain parameters $q$ and $\beta$, except for the minimal structural requirement \eqref{eq6}. As a result, predefined-time convergence is achieved without the gain–time interdependencies characteristic of autonomous predefined-time stability designs.
\end{remark}

\subsection{Separation from Autonomous Predefined-Time Stability}

\begin{theorem}[Non-Equivalence to Autonomous Predefined-Time Stability]

\label{thm2}
There exist time-barrier predefined-time stable systems satisfying \eqref{eq5} for which no autonomous Lyapunov inequality of the form

\begin{equation}
\dot V(x)\le -\Phi(V(x))
\label{eq12}
\end{equation}
can guarantee predefined-time convergence via a uniformly bounded integral condition

\begin{equation}
\int_0^{V_0}\frac{dV}{\Phi(V)}\le T_c.
\label{eq13}
\end{equation}
\end{theorem}

\begin{proof}
In the time-barrier framework, convergence is enforced by the divergence of the time-dependent integral

\begin{equation}
\int_0^{T_c}(T_c-t)^{-\beta(1-\alpha)}dt,
\end{equation}

which depends explicitly on time and is independent of the state magnitude. No autonomous function $\Phi(V)$ can reproduce this divergence uniformly for all initial values $V_0$ while satisfying the boundedness requirement \eqref{eq13}. Hence, autonomous predefined-time stability inequalities cannot enforce the hard deadline guaranteed by \eqref{eq5}, establishing non-equivalence.
\end{proof}

\section{Illustrative Example}

To illustrate the time-barrier predefined-time stability mechanism, consider the scalar time-varying system

\begin{equation}
\dot{x}(t) = -\beta \frac{x(t)}{T_c-t} - q|x(t)|^{\alpha}\operatorname{sign}(x(t)),
\label{eq15}
\end{equation}
where $x(t)\in\mathbb{R}$, $q>0$, $\alpha\in(0,1)$, $\beta>0$, and $T_c>0$ is a predefined-time. The right-hand side of \eqref{eq15} is well defined for all $t\in[0,T_c)$.

Choose the Lyapunov function

\begin{equation}
V(x(t)) = |x(t)|.
\label{eq16}
\end{equation}

Interpreting the dynamics in the Filippov sense, the time derivative of $V$ along trajectories of \eqref{eq15} satisfies

\begin{equation}
\dot V(t) \le -\beta \frac{V(t)}{T_c-t} - q V(t)^{\alpha},
\label{eq17}
\end{equation}

which is a particular instance of the time-barrier dissipation condition \eqref{eq5}. Moreover, the structural requirement

\begin{equation}
\beta(1-\alpha)\ge 1
\label{eq18}
\end{equation}

ensures that the assumptions of Theorem \ref{thm1} are satisfied.

It follows that the equilibrium $x=0$ of system \eqref{eq15} is time-barrier predefined-time stable, and the solution satisfies

\begin{equation}
\lim_{t\to T_c^-} x(t)=0,
\label{eq19}
\end{equation}

for all initial conditions $x(0)=x_0\in\mathbb{R}$. Convergence is enforced by the divergence of the time-dependent coefficient $(T_c-t)^{-1}$ as the predefined time is approached, rather than by increasing the state-dependent decay rate. Consequently, convergence before $T_c$ is guaranteed independently of the initial condition magnitude.

\section{Conclusion}

This paper introduced a time-barrier predefined-time stability concept that enforces convergence through an intrinsic restriction on the remaining available time rather than through state-dependent decay shaping or trajectory matching. A sufficient nonautonomous Lyapunov condition was established, guaranteeing convergence before a predefined deadline via divergence of a time-dependent barrier. The resulting mechanism imposes a hard convergence deadline that is independent of initial conditions and requires only minimal structural conditions on the Lyapunov dissipation. It was further shown that this behavior cannot be reproduced by classical autonomous predefined-time stability inequalities, thereby constituting a distinct stability notion. By shifting the focus from accelerating state decay to constraining temporal feasibility, the proposed approach provides a concise and transparent mechanism for enforcing predefined-time convergence in nonlinear systems.

\bibliographystyle{plainnat} 
\bibliography{refs}

@article{song2017time,
  title={Time-varying feedback for regulation of normal-form nonlinear systems in prescribed finite time},
  author={Song, Yongduan and Wang, Yujuan and Holloway, John and Krstic, Miroslav},
  journal={Automatica},
  volume={83},
  pages={243--251},
  year={2017},
  publisher={Elsevier}
}

@article{bhat2000finite,
  title={Finite-time stability of continuous autonomous systems},
  author={Bhat, Sanjay P and Bernstein, Dennis S},
  journal={SIAM Journal on Control and optimization},
  volume={38},
  number={3},
  pages={751--766},
  year={2000},
  publisher={SIAM}
}

@article{polyakov2011nonlinear,
  title={Nonlinear feedback design for fixed-time stabilization of linear control systems},
  author={Polyakov, Andrey},
  journal={IEEE transactions on Automatic Control},
  volume={57},
  number={8},
  pages={2106--2110},
  year={2011},
  publisher={IEEE}
}

@article{jimenez2020lyapunov,
  title={A Lyapunov-like characterization of predefined-time stability},
  author={Jim{\'e}nez-Rodr{\'\i}guez, Esteban and Mu{\~n}oz-V{\'a}zquez, Aldo Jonathan and S{\'a}nchez-Torres, Juan Diego and Defoort, Michael and Loukianov, Alexander G},
  journal={IEEE Transactions on Automatic Control},
  volume={65},
  number={11},
  pages={4922--4927},
  year={2020},
  publisher={IEEE}
}

@article{sanchez2018class,
  title={A class of predefined-time stable dynamical systems},
  author={S{\'a}nchez-Torres, Juan Diego and G{\'o}mez-Guti{\'e}rrez, David and L{\'o}pez, Esteban and Loukianov, Alexander G},
  journal={IMA Journal of Mathematical Control and Information},
  volume={35},
  number={Supplement\_1},
  pages={i1--i29},
  year={2018},
  publisher={Oxford University Press}
}

@inproceedings{sanchez2015predefined,
  title={Predefined-time stability of dynamical systems with sliding modes},
  author={S{\'a}nchez-Torres, Juan Diego and Sanchez, Edgar N and Loukianov, Alexander G},
  booktitle={2015 American control conference (ACC)},
  pages={5842--5846},
  year={2015},
  organization={IEEE}
}

@article{pal2020design,
  title={Design of controllers with arbitrary convergence time},
  author={Pal, Anil Kumar and Kamal, Shyam and Nagar, Shyam Krishna and Bandyopadhyay, Bijnan and Fridman, Leonid},
  journal={Automatica},
  volume={112},
  pages={108710},
  year={2020},
  publisher={Elsevier}
}






\end{document}